\documentclass[12pt,reqno,tbtags]{article}
\usepackage{amsmath}
\usepackage{float}
\usepackage{graphicx}
\usepackage{amsthm}
\usepackage{amssymb}
\usepackage{framed}
\usepackage{tikz}
\usepackage{enumerate}
\usetikzlibrary{arrows}

\usepackage{caption,sansmath}
\DeclareCaptionFont{sansmath}{\sansmath}
\usepackage{todonotes}
\usepackage{comment}

\newtheorem{theorem}{Theorem}[section]
\newtheorem{lem}[theorem]{Lemma}
\newtheorem{prop}[theorem]{Proposition}
\newtheorem{corollary}[theorem]{Corollary}

\newtheorem{conjecture}{Conjecture}[section]

\newtheorem{claim}{Claim}

\begin{document}

\title{On a conjecture of Mohar concerning\\ Kempe equivalence of regular graphs}

\author{
	Marthe Bonamy\thanks{CNRS, LaBRI, Bordeaux, France. \newline email: \texttt{marthe.bonamy@labri.fr}}, \quad
    Nicolas Bousquet\thanks{CNRS, G-SCOP, Grenoble, France. \newline email: \texttt{nicolas.bousquet@ec-lyon.fr}},\\
	Carl Feghali\thanks{School of Engineering and Computing Sciences, Durham University, United Kingdom. \newline emails: \texttt{\{carl.feghali,matthew.johnson2\}@durham.ac.uk}}, \quad
	Matthew Johnson\footnotemark[3]
}
\date{}

\maketitle

\begin{abstract}
Let $G$ be a graph with a vertex colouring $\alpha$.  Let $a$ and $b$ be two colours.  Then a connected component of the subgraph induced by those vertices coloured either $a$ or $b$ is known as a Kempe chain.  A colouring of $G$ obtained from $\alpha$ by swapping the colours on the vertices of a Kempe chain is said to have been obtained by a Kempe change.  Two colourings of $G$ are Kempe equivalent if one can be obtained from the other by a sequence of Kempe changes.

A conjecture of Mohar (2007) asserts that, for $k \geq 3$, all $k$-colourings of a $k$-regular graph that is not complete are Kempe equivalent.  It was later shown that all $3$-colourings of a cubic graph that is neither~$K_4$ nor the triangular prism are Kempe equivalent. In this paper, we prove that the conjecture holds for each $k\geq 4$.  We also report the implications of this result on the  validity of the Wang-Swendsen-Koteck\'{y} algorithm for the antiferromagnetic Potts model at zero-temperature.
\end{abstract}

\section{Introduction}

Let $G = (V, E)$ denote a simple undirected graph.  Let $k$ be a positive integer.   Then a \emph{$k$-colouring} of~$G$ is a function $\alpha: V \rightarrow \{1, \dots, k\}$ such that for each edge $uv \in E$,  $\alpha(u) \not= \alpha(v)$.  We call $\{1, \dots, k\}$ the set of colours and refer to~$\alpha(u)$ as the colour of the vertex $u$. For a colouring~$\alpha$ and colours~$a$ and~$b$, $G_{\alpha}(a, b)$ is the subgraph of $G$ induced by vertices with colour~$a$ or~$b$. A connected component of $G_{\alpha}(a, b)$ is known as an \emph{$(a,b)$-component} of $G$ and~$\alpha$.  These components are also referred to as \emph{Kempe chains}. If a colouring $\beta$ is obtained from a colouring $\alpha$ by exchanging the colours $a$ and~$b$ on the vertices of an $(a,b)$-component of $G$ and $\alpha$, then $\beta$ is said to have been obtained from~$\alpha$ by a \emph{Kempe change}.  A pair of colourings are \emph{Kempe equivalent} if each can be obtained from the other by a sequence of Kempe changes.  A set of Kempe equivalent colourings is called a \emph{Kempe class}.

A graph is \emph{$k$-regular} if every vertex has degree $k$. By Brooks' Theorem~\cite{brooks}, a $k$-regular connected graph has a $k$-colouring unless it is a complete graph or a cycle on an odd number of vertices.  Mohar conjectured that for all other $k$-regular graphs, the set of $k$-colourings form a Kempe class~\cite{mohar1}; that is, any possible colouring can be obtained from an initial colouring by a series of Kempe changes.  The first non-trivial case is when $k=3$ and van den Heuvel~\cite{heuvel} showed that there is a counterexample to the conjecture, the triangular prism (the graph obtained by joining the vertices of two vertex-disjoint triangles by a perfect matching, see Figure~\ref{fig:prism} below and the discussion at the end of this section). 

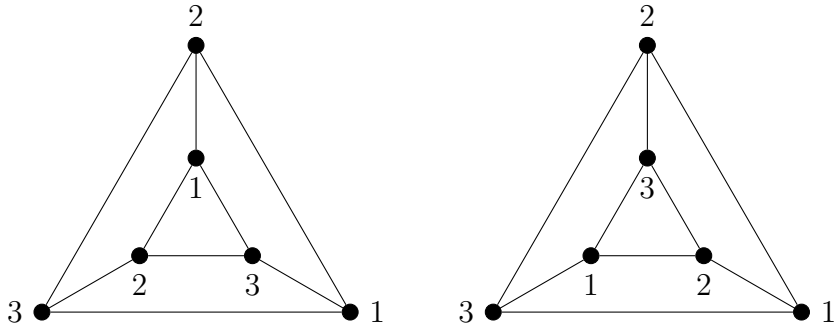
\begin{figure}[!h]
\center
\begin{tikzpicture}
    \tikzstyle{whitenode}=[draw,circle,fill=white,minimum size=9pt,inner sep=0pt]
    \tikzstyle{blacknode}=[draw,circle,fill=black,minimum size=6pt,inner sep=0pt]
\draw (0,0) node[blacknode] (c) [label=-90:$2$] {} 
 ++(0:1.5cm) node[blacknode] (b) [label=-90:$3$] {}
 ++(120:1.5cm) node[blacknode] (a) [label=-90:$1$] {};
 
 \draw (a)
-- ++(90:1.5cm) node[blacknode] (a2) [label=90:$2$] {};
  \draw (b)
-- ++(-30:1.5cm) node[blacknode] (b2) [label=0:$1$] {};
  \draw (c)
-- ++(-150:1.5cm) node[blacknode] (c2) [label=180:$3$] {};
 
 \draw (a) -- (b);
 \draw (c) -- (b);
 \draw (a) -- (c);
 \draw (a2) -- (b2);
 \draw (c2) -- (b2);
 \draw (a2) -- (c2);
 
 \draw (6,0) node[blacknode] (c) [label=-90:$1$] {} 
 ++(0:1.5cm) node[blacknode] (b) [label=-90:$2$] {}
 ++(120:1.5cm) node[blacknode] (a) [label=-90:$3$] {};
 
 \draw (a)
-- ++(90:1.5cm) node[blacknode] (a2) [label=90:$2$] {};
  \draw (b)
-- ++(-30:1.5cm) node[blacknode] (b2) [label=0:$1$] {};
  \draw (c)
-- ++(-150:1.5cm) node[blacknode] (c2) [label=180:$3$] {};
 
 \draw (a) -- (b);
 \draw (c) -- (b);
 \draw (a) -- (c);
 \draw (a2) -- (b2);
 \draw (c2) -- (b2);
 \draw (a2) -- (c2);

\end{tikzpicture}
\caption{The triangular prism with two non-Kempe-equivalent $3$-colourings.}
\label{fig:prism}
\end{figure}

 Recently Feghali et~al.~\cite{FJP15} showed that, for \mbox{$k=3$}, this is the only counterexample: for a non-complete 3-regular connected graph~$G$, the 3-colourings of $G$ form a Kempe class \emph{unless} $G$ is the triangular prism.  In this paper, we affirm the conjecture for larger~$k$.

\begin{theorem} \label{thm:main}
Let $k \geq 4$ be a positive integer.  If $G$ is a connected non-complete $k$-regular graph,  then the set of $k$-colourings of $G$ is a Kempe class.
\end{theorem}

We consider only connected graphs as other graphs can be considered componentwise.  Notice that we need not have included the condition that~$G$ is not complete since one can say that if a graph has no $k$-colourings, then this set of colourings (the empty set) is a Kempe class, but it is neater to exclude this case.  We will note in Corollary~\ref{cor:final} that this result about $k$-regular graphs can easily be extended to an analogous result on graphs of maximum degree~$k$.
 
Let us describe another way to think of Theorem~\ref{thm:main}.  Let $C_G^k$ be the set of $k$-colourings of a graph $G$.  Let $\mathcal{K}_k(G)$ be the graph that has vertex set $C_G^k$ and an edge between two vertices $\alpha$ and~$\beta$ whenever the colouring~$\beta$ can be obtained from $\alpha$ by a single Kempe change.  Theorem~\ref{thm:main} states that, for $k \geq 4$, for any connected non-complete $k$-regular graph $G$, $\mathcal{K}_k(G)$ is connected.

We might call $\mathcal{K}_k(G)$ a \emph{solution graph}; it represents all possible solutions to the problem of finding a $k$-colouring of $G$.  Or we can call it the \mbox{\emph{reconfiguration graph}} of $k$-colourings of $G$ and refer to Kempe changes as reconfiguration steps.  In fact, reconfiguration graphs of $k$-colourings have been much studied when the edge relation is defined by the alternative reconfiguration step of \emph{trivial} Kempe changes. A Kempe chain is trivial if it contains a single vertex $v$ and the corresponding Kempe change alters the colour of only~$v$, and so pairs of colourings are connected in these reconfiguration graphs if they disagree on only one vertex.  These graphs were introduced in~\cite{CHJ06}.  Much work on these graphs has focussed on (the computational complexity of) deciding whether or not the reconfiguration graph is connected~\cite{CHJ06a}, on deciding whether a given pair of colourings belong to the same connected component~\cite{BC09,BM14, CHJ06b,JKKPP14}, and on the diameter of the reconfiguration graph or its components~\cite{BB13,BJLPP14,BP14,FJP14}.  Similar work has been done for reconfiguration graphs for search problems other than graph colouring; see, for example, the survey of van den Heuvel~\cite{heuvel}.  

Reconfiguration graphs defined by Kempe changes have received less attention. Kempe changes were introduced in 1879 by Kempe in his attempted proof of the Four Colour Theorem~\cite{kempe79}.  Though this proof was fallacious, the Kempe change technique has proved useful in, for example, the proof of the Five Colour Theorem and a short proof of Brooks' Theorem~\cite{vizing}.  We briefly review the purely graph theoretical studies of Kempe equivalence.  Fisk~\cite{fisk} showed in 1977 that the $4$-colourings of a Eulerian triangulation of the plane are a Kempe class.  This was generalized both by Meyniel~\cite{meyniel1}, who showed that the $5$-colourings of a planar graph are a Kempe class, and  by Mohar~\cite{mohar1}, who proved that the $k$-colourings of a planar graph~$G$ are a Kempe class if $k > \chi(G)$. The former result was further extended by Las Vergnas and Meyniel~\cite{meyniel3} who showed that the $5$-colourings of a $K_5$-minor free graph are a Kempe class. Bertschi~\cite{marc} proved that the $k$-colourings of a perfectly contractile graph are a Kempe class (a graph is perfectly contractile if, from each of its induced subgraphs, a complete graph can be obtained by repeatedly contracting pairs of vertices that have no odd-length induced path between them).  The Kempe equivalence of edge-colourings has also been investigated~\cite{ruth, mohar2, mohar1}.

Applications of the Kempe change technique can be found in theoretical physics~\cite{mohar3, mohar4, sokal, wang1, wang2} and in the study of Markov chains~\cite{vigoda} and timetables~\cite{rolf}. In particular, our result implies that the Wang-Swendsen-Koteck\'{y} (WSK) algorithm  for the zero-temperature $q$-state Potts antiferromagnetic model is ergodic on the triangular lattice with periodic boundary conditions (that is, embedded on the torus) for $q = 6$, and on the Kagom\'e lattice with periodic boundary conditions for $q = 4$, thus answering some of the questions raised in~\cite{mohar3, mohar4}.  We discuss this in Section~\ref{section:app}.

In Section~\ref{section:pre} we introduce some useful lemmas.  In Section~\ref{section:final}, we prove Theorem~\ref{thm:main}. In Section~\ref{sec:final}, we present a corollary and some closing remarks. We conclude this section with some final comments on our investigations towards proving Theorem~\ref{thm:main}.  We now know that, for $k \geq 3$, the only non-complete connected $k$-regular graph whose $k$-colourings are not a Kempe class is the triangular prism.  Let us explain why its 3-colourings are not all Kempe equivalent. Notice from Figure~\ref{fig:prism} that no Kempe change modifies the colour partition.  Thus the two 3-colourings illustrated are not Kempe equivalent as they are not the same up to colour permutation.

So one might have hoped to find a counterexample to Theorem~\ref{thm:main} by finding, for some $k \geq 4$, a connected non-complete $k$-regular graph with a $k$-colouring such that all Kempe changes maintain the colour partition.  However, it is not hard to convince oneself that such a graph does not exist. Indeed, let us consider such a graph $G$ and $k$-colouring $\alpha$ and obtain a contradiction.  As $G$ has more than $k$ vertices some colour $a$ appears on more than one vertex.  If a colour $b$ does not appear on any vertex, then changing the colour of one vertex from $a$ to $b$ gives a colouring with a different partition.  And, if $b$ appears on only one vertex $u$, then changing the colour of a vertex not adjacent to $u$ to $b$ again changes the partition.  So every colour appears on at least two vertices.  If, for any vertex $u$, there is a colour other than~$\alpha(u)$ that does not appear in its neighbourhood, then another trivial Kempe change gives a colouring with a different partition; so on the $k$ vertices in the neighbourhood of $u$, one colour appears twice and every other colour but $\alpha(u)$ appears once.  For every pair of colours $a$ and $b$, $G_{\alpha}(a,b)$ is connected (else a Kempe change of one component gives a different partition). Since $k\geq 4$, there are at least four distinct colours $a$, $b$, $c$ and~$d$.  As every vertex in $G_{\alpha}(a,b)$ has degree 1 or 2, it is either a path or a cycle.  As there are at least two vertices coloured $a$, there is a vertex $u$ coloured $c$ that has degree 2 in $G_{\alpha}(a,c)$.  Similarly, there is a vertex $v$ coloured~$c$ that has degree~2 in $G_{\alpha}(b,c)$; clearly $u \neq v$.  Notice that $u$ and $v$ must both have degree~1 in $G_{\alpha}(c,d)$; that is $G_{\alpha}(c,d)$ is a path whose endvertices are both coloured $c$.  But, by the same argument, $G_{\alpha}(c,d)$ is a path whose endvertices are both coloured~$d$.  This contradiction proves that such a $G$ and $\alpha$ cannot exist.

\section{Ergodicity of the WSK algorithm}\label{section:app}

The $q$-state Potts models~\cite{potts1, potts2, potts3} are among the simplest and most studied models in statistical mechanics with various applications from the theory of critical phenomena to condensed-matter systems. The model uses a finite graph $G =(V, E)$ where each vertex $v \in V$ is assigned a \emph{spin} $\sigma(v) \in \{1, \dots, q\}$ (that is, the spins provide a not necessarily proper vertex-colouring of~$G$).  Furthermore, the $q$-state Potts models are dynamic models where the spin state of a vertex can be modified over time with probabilities depending on the spin states of adjacent vertices. There exist two main Potts models. In the \emph{ferromagnetic Potts model}, the state of a spin is attracted to the spin states of adjacent vertices (equilibria correspond to monochromatic graph colourings).  We are concerned with the \emph{antiferromagnetic Potts model} where the state of a spin is repelled by the spin states of adjacent vertices, which means that every spin ``tries'' to achieve a state that is distinct from its neighbours. 

The Wang-Swendsen-Koteck\'{y} (WSK) non-local cluster dynamics~\cite{wang1, wang2} is  a popular (Markov chain) Monte Carlo algorithm used to simulate the antiferromagnetic Potts model. In order to be \emph{valid} (in other words, to work properly), Monte Carlo simulations require \emph{ergodicity}, which means that there must be a positive probability of transforming each configuration into any other.  The zero temperature hypothesis is that, at that temperature, any two adjacent vertices must have distinct spins.  Therefore the spin function corresponds to a proper $q$-colouring of $G$, and the ergodicity of the Markov chain holds if and only if the set of $q$-colourings is a Kempe class. Since this does not hold for graphs in general, the statistical mechanics community has studied the ergodicity of the Markov chain on special graphs, especially highly structured graphs that can be embedded on surfaces.  Amongst them, triangular lattices and Kagom\'{e} lattices with boundary conditions have received considerable attention. Both these graph classes can be defined in terms of toroidal grid graphs (graphs that are the Cartesian product of two cycles).  A triangular lattice with boundary conditions (abbreviated to triangular lattice) is formed from a toroidal grid graph by adding a single ``diagonal'' edge to each face of the grid in such a way that all the diagonals are parallel.  Thus the triangular lattice is 6-regular and every face is a triangle. A Kagom{\'e} lattice with boundary conditions (abbreviated to  Kagom{\'e} lattice) is obtained from a toroidal grid graph by first subdividing each edge.  Each face of the grid now contains four midpoints (vertices of degree two) on its, let us say, north, east, south and west sides.  Then the north and east midpoints, and the south and west midpoints, are joined by an edge.  Thus the Kagom\'e lattice is 4-regular, every face is a triangle or a hexagon and every edge belongs to exactly one triangle and one hexagon.  (One can also think of the Kagom\'e lattice as the line graph of a honeycomb embedded on the torus.)  The two lattices are illustrated in Figure~\ref{fig:lattice}.  In Theorem~\ref{thm:lattice}, we summarise what is currently known about the validity of the WSK algorithm on these lattices, including new results implied by Theorem~\ref{thm:main}.  Note that in this work we do not exhibit any bound on the mixing time (informally speaking, the speed of convergence) of the WSK algorithm on these instances. Lower bounds on the recolouring diameter (see, for example,~\cite{BJLPP14}) imply lower bounds on the mixing time on slightly different models. However, upper bounds on the recolouring diameter (see~\cite{BB13,BJLPP14,BP14, CHJ06b}) have no implications for the mixing time, which suggests an interesting direction of research.\newline

We first state a lemma of Las Vergnas and Meyniel~\cite{meyniel3} that we will need to prove the theorem.

\tikzstyle{vertex}=[circle,draw=black, minimum size=4pt, inner sep=0pt]

\newcommand{\row}[0]{
    \foreach \pos in {(0,0), (0.5,0), (1,0), (1.5,0), (2,0), (2.5,0), (3,0), (3.5,0), (4,0), (4.5,0)}
        \node[vertex, fill=black] () at \pos {};
               \path[draw, -, black] (-0.2,0) --  (4.7,0);
} 

\newcommand{\mrow}[0]{
    \foreach \pos in {(0,0), (0.5,0), (1,0), (1.5,0), (2,0), (2.5,0), (3,0), (3.5,0), (4,0), (4.5,0)}
        \node[vertex, fill=white] () at \pos {};
} 

\newcommand{\shortmrow}[0]{
    \foreach \pos in {(0.25,0), (0.75,0), (1.25,0), (1.75,0), (2.25,0), (2.75,0), (3.25,0), (3.75,0), (4.25,0)}
        \node[vertex, fill=white] () at \pos {};
} 

\begin{figure}  
\begin{center}
\begin{tikzpicture}[scale=1.15]


\begin{scope} \row \end{scope}
\begin{scope}[yshift=0.5cm] \row \end{scope}
\begin{scope}[yshift=1cm] \row \end{scope}
\begin{scope}[yshift=1.5cm] \row \end{scope}
\begin{scope}[yshift=2cm] \row \end{scope}
\begin{scope}[yshift=2.5cm] \row \end{scope}
\begin{scope}[yshift=3cm] \row \end{scope}
\begin{scope}[yshift=3.5cm] \row \end{scope}
\begin{scope}[yshift=4cm] \row \end{scope}
\begin{scope}[yshift=4.5cm] \row \end{scope}

\path[draw, -, black] (0,-0.2) --  (0,4.7);
\path[draw, -, black] (0.5,-0.2) --  (0.5,4.7);
\path[draw, -, black] (1,-0.2) --  (1,4.7);
\path[draw, -, black] (1.5,-0.2) --  (1.5,4.7);
\path[draw, -, black] (2,-0.2) --  (2,4.7);
\path[draw, -, black] (2.5,-0.2) --  (2.5,4.7);
\path[draw, -, black] (3,-0.2) --  (3,4.7);
\path[draw, -, black] (3.5,-0.2) --  (3.5,4.7);
\path[draw, -, black] (4,-0.2) --  (4,4.7);
\path[draw, -, black] (4.5,-0.2) --  (4.5,4.7);

\path[draw, -, black] (0.2,-0.2) --  (-0.2,0.2);
\path[draw, -, black] (0.7,-0.2) --  (-0.2,0.7);
\path[draw, -, black] (1.2,-0.2) --  (-0.2,1.2);
\path[draw, -, black] (1.7,-0.2) --  (-0.2,1.7);
\path[draw, -, black] (2.2,-0.2) --  (-0.2,2.2);
\path[draw, -, black] (2.7,-0.2) --  (-0.2,2.7);
\path[draw, -, black] (3.2,-0.2) --  (-0.2,3.2);
\path[draw, -, black] (3.7,-0.2) --  (-0.2,3.7);
\path[draw, -, black] (4.2,-0.2) --  (-0.2,4.2);
\path[draw, -, black] (4.7,-0.2) --  (-0.2,4.7);

\path[draw, -, black] (4.7,0.3) --  (0.3,4.7);
\path[draw, -, black] (4.7,0.8) --  (0.8,4.7);
\path[draw, -, black] (4.7,1.3) --  (1.3,4.7);
\path[draw, -, black] (4.7,1.8) --  (1.8,4.7);
\path[draw, -, black] (4.7,2.3) --  (2.3,4.7);
\path[draw, -, black] (4.7,2.8) --  (2.8,4.7);
\path[draw, -, black] (4.7,3.3) --  (3.3,4.7);
\path[draw, -, black] (4.7,3.8) --  (3.8,4.7);
\path[draw, -, black] (4.7,4.3) --  (4.3,4.7);


\begin{scope}[xshift=6.5cm]

\begin{scope} \row \end{scope}
\begin{scope}[yshift=0.5cm] \row \end{scope}
\begin{scope}[yshift=1cm] \row \end{scope}
\begin{scope}[yshift=1.5cm] \row \end{scope}
\begin{scope}[yshift=2cm] \row \end{scope}
\begin{scope}[yshift=2.5cm] \row \end{scope}
\begin{scope}[yshift=3cm] \row \end{scope}
\begin{scope}[yshift=3.5cm] \row \end{scope}
\begin{scope}[yshift=4cm] \row \end{scope}
\begin{scope}[yshift=4.5cm] \row \end{scope}

\path[draw, -, black] (0,-0.2) --  (0,4.7);
\path[draw, -, black] (0.5,-0.2) --  (0.5,4.7);
\path[draw, -, black] (1,-0.2) --  (1,4.7);
\path[draw, -, black] (1.5,-0.2) --  (1.5,4.7);
\path[draw, -, black] (2,-0.2) --  (2,4.7);
\path[draw, -, black] (2.5,-0.2) --  (2.5,4.7);
\path[draw, -, black] (3,-0.2) --  (3,4.7);
\path[draw, -, black] (3.5,-0.2) --  (3.5,4.7);
\path[draw, -, black] (4,-0.2) --  (4,4.7);
\path[draw, -, black] (4.5,-0.2) --  (4.5,4.7);

\path[draw, -, black] (0.45,-0.2) --  (-0.2,0.45);
\path[draw, -, black] (0.95,-0.2) --  (-0.2,0.95);
\path[draw, -, black] (1.45,-0.2) --  (-0.2,1.45);
\path[draw, -, black] (1.95,-0.2) --  (-0.2,1.95);
\path[draw, -, black] (2.45,-0.2) --  (-0.2,2.45);
\path[draw, -, black] (2.95,-0.2) --  (-0.2,2.95);
\path[draw, -, black] (3.45,-0.2) --  (-0.2,3.45);
\path[draw, -, black] (3.95,-0.2) --  (-0.2,3.95);
\path[draw, -, black] (4.45,-0.2) --  (-0.2,4.45);

\path[draw, -, black] (4.7,0.05) --  (0.05,4.7);
\path[draw, -, black] (4.7,0.55) --  (0.55,4.7);
\path[draw, -, black] (4.7,1.05) --  (1.05,4.7);
\path[draw, -, black] (4.7,1.55) --  (1.55,4.7);
\path[draw, -, black] (4.7,2.05) --  (2.05,4.7);
\path[draw, -, black] (4.7,2.55) --  (2.55,4.7);
\path[draw, -, black] (4.7,3.05) --  (3.05,4.7);
\path[draw, -, black] (4.7,3.55) --  (3.55,4.7);
\path[draw, -, black] (4.7,4.05) --  (4.05,4.7);

\begin{scope}[yshift=0.25cm] \mrow \end{scope}
\begin{scope}[yshift=0.75cm] \mrow \end{scope}
\begin{scope}[yshift=1.25cm] \mrow \end{scope}
\begin{scope}[yshift=1.75cm] \mrow \end{scope}
\begin{scope}[yshift=2.25cm] \mrow \end{scope}
\begin{scope}[yshift=2.75cm] \mrow \end{scope}
\begin{scope}[yshift=3.25cm] \mrow \end{scope}
\begin{scope}[yshift=3.75cm] \mrow \end{scope}
\begin{scope}[yshift=4.25cm] \mrow \end{scope}

\begin{scope} \shortmrow \end{scope}
\begin{scope}[yshift=0.5cm] \shortmrow \end{scope}
\begin{scope}[yshift=1cm] \shortmrow \end{scope}
\begin{scope}[yshift=1.5cm] \shortmrow \end{scope}
\begin{scope}[yshift=2cm] \shortmrow \end{scope}
\begin{scope}[yshift=2.5cm] \shortmrow \end{scope}
\begin{scope}[yshift=3cm] \shortmrow \end{scope}
\begin{scope}[yshift=3.5cm] \shortmrow \end{scope}
\begin{scope}[yshift=4cm] \shortmrow \end{scope}
\begin{scope}[yshift=4.5cm] \shortmrow \end{scope}

\end{scope}

\end{tikzpicture}
\end{center}

\caption{Portions of a triangular lattice and a Kagom\'e lattice}
\label{fig:lattice}
\end{figure}
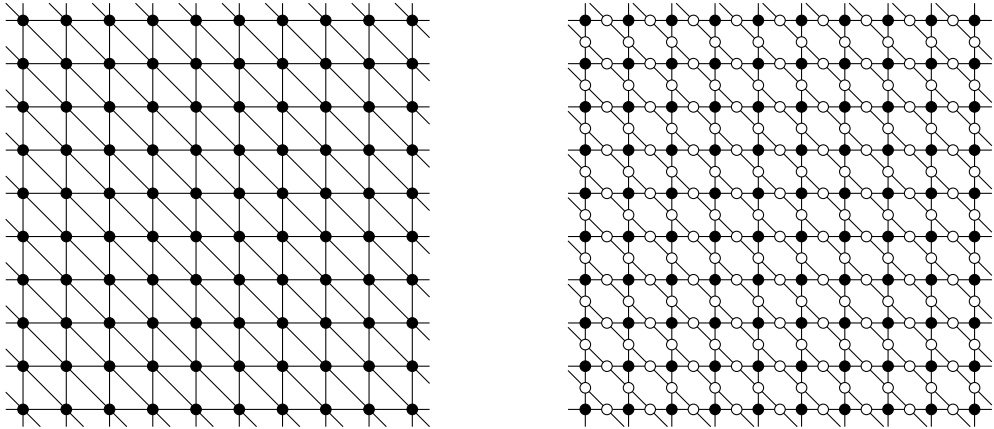
 
\begin{lem}[\cite{meyniel3}]\label{lem:degeneratekempe}
Let $d$ and $k$ be positive integers, $d \geq 0$, $k \geq d+1$. If $G$ is a $d$-degenerate graph, then $C^k_G$ is a Kempe class. 
\end{lem}

\begin{theorem} \label{thm:lattice}
The WSK algorithm for $q$-colourings of the triangular lattice  is valid if $q \geq 6$ and is not valid if $q \leq 4$. 
The WSK algorithm for $q$-colourings of the Kagom\'{e} lattice is valid if $q \geq 4$ and is not valid if $q \leq 3$.
\end{theorem} 
 
 \begin{proof}
For the triangular lattice,  Mohar and Salas showed in~\cite{mohar3} that the chain is not ergodic when $q \leq 4$. Theorem~\ref{thm:main} ensures that when $q = 6$, as the triangular lattice is $6$-regular, the WSK Markov chain is ergodic, and then the WSK algorithm is valid. For larger values of $q$, Lemma~\ref{lem:degeneratekempe} ensures that the chain is also ergodic.

For the Kagom\'{e} lattice,  Mohar and Salas proved in~\cite{mohar4} that the chain is not ergodic when $q \leq 3$. Theorem~\ref{thm:main} ensures that when $q = 4$, as the Kagom\'{e} lattice is $4$-regular, the WSK` Markov chain is ergodic and then the WSK algorithm is valid. And, again, for larger values of $q$, we use Lemma~\ref{lem:degeneratekempe}.
 \end{proof}
 
We observe that this leaves the single open case of a triangular lattice with $q = 5$. 

We note that in this section, we have discussed the validity of the WSK algorithm on lattices with boundary conditions. The validity of the WSK algorithm on lattices with free conditions (that is, lattices embedded in the plane) has already been studied by Mohar and Salas~\cite{mohar3, mohar4}.

\section{Preliminaries}\label{section:pre}

Let $S$ be a subset of the vertex set of a graph $G$.  Then $G[S]$ denotes the subgraph of $G$ induced by $S$.  Let $d$ be a positive integer. Then a \emph{$d$-elimination ordering} of the vertices of $G$ is an ordering such that each vertex is adjacent to at most $d$ vertices later in the ordering.   We say that the ordering \emph{ends in $S$} if the vertices of $S$ are later in the ordering than all other vertices.  A graph is \emph{$d$-degenerate} if there is a $d$-elimination ordering of its vertices, or, equivalently, if every induced subgraph has a vertex of degree at most $d$.  From these definitions we immediately obtain:

\begin{lem}\label{rem:EliminationOrdering}
Let $d$ be a positive integer.  Let $G$ be a graph, and let~$S$ be a subset of the vertices of $G$. If $G$ admits a $d$-elimination ordering that ends in~$S$, then any $(d+1)$-colouring of $G[S]$ can be extended to a $(d+1)$-colouring of $G$.
\end{lem}
 
Let us refine this in a way that will prove useful. 
 
\begin{lem}\label{lem:elimination}
Let $k$ be a positive integer.  Let $G=(V,E)$ be a graph, and let~$S \subseteq V$, $|S| \leq k$, be a subset of the vertex set.  Suppose that $G[V \setminus S]$ is connected, that the vertices of $V \setminus S$ each have degree at most $k$ in $G$, and there is a vertex $x \in V \setminus S$ of degree at most $k-1$ in $G$.   Then any $k$-colouring of $G[S]$ can be extended to a $k$-colouring of~$G$.
\end{lem}

\begin{proof}
Let the vertices of $V \setminus S$ be ordered according to the order in which they are found by a breadth-first search from $x$ (and certainly the search finds all these vertices as $G[V \setminus S]$ is connected).  Append the vertices of $S$ to this ordering.  This is certainly a $(k-1)$-elimination ordering of $G$ since~$x$ has at most $k-1$ neighbours in total, every other vertex in $V \setminus S$ has at most~$k$ neighbours, but at least one --- the vertex from which it was discovered during the breadth-first search --- is earlier in the ordering, and each vertex of $S$ is followed in the ordering only by other vertices of $S$ of which there are at most $k-1$.  So, by Lemma~\ref{rem:EliminationOrdering} with $d=k-1$, the $k$-colouring of $S$ can be extended to~$G$.
\end{proof} 
 
We need some known results. 

\begin{lem}[\cite{meyniel3}]\label{lem:cliqueintersection}
Let $k$ be a positive integer. Let $G_1, G_2$ be two graphs such that $G_1 \cap G_2$ is complete. If both $C^k_{G_1}$ and $C^k_{G_2}$ are Kempe classes, then~$C^k_{G_1 \cup G_2}$ is a Kempe class.
\end{lem}

We \emph{identify} two non-adjacent vertices $u$ and $v$ in a graph~$G$ if we replace them by a new vertex adjacent to all neighbours of either $u$ or $v$ (or both) in~$G$.  The graph obtained is denoted $G_{u+v}$.  In the proof of Theorem~\ref{thm:main}, we will often think about $G_{u+v}$ when reasoning about the colourings of $G$.  Let~$C^k_G(u,v)$ denote the colourings of~$G$ for which $u$ and $v$ are coloured alike.  We note that there is an obvious bijection between~$C^k_G(u,v)$ and $C^k_{G_{u+v}}$.

\begin{lem}[\cite{FJP15}]\label{lem:matching}
Let $k \geq 3$ be a positive integer. Let $G$ be a $3$-connected graph. Let $u$ and $v$ be non-adjacent vertices of $G$ with a common neighbour.  Then, if $C^k_G(u,v)$ is not empty, it is a Kempe class.
\end{lem}

\noindent We observe that Lemma~\ref{lem:matching} is \cite[Lemma 9]{FJP15} (in fact, that states, equivalently, that any two colourings of~$C^k_G(u,v)$ are Kempe equivalent). The proof in~\cite{FJP15} first establishes the following statement which it is useful to state explicitly.

\begin{lem}[\cite{FJP15}]\label{lem:degenerate}
Let $k\geq 3$ be a positive integer.  Let $G$ be a 3-connected graph of maximum degree~$k$.  Let $u$ and $v$ be non-adjacent vertices of $G$ with a common neighbour. Then $G_{u+v}$  is $(k-1)$-degenerate. 
\end{lem}

A list assignment  of a graph $G=(V,E)$ is a function $L$ with domain $V$ such that, for each vertex $u \in V$, $L(u)$ is a set of colours. We say that $G$ is \emph{$L$-colourable} if there is a colouring of $G$ where every vertex~$u$ is coloured with a colour of $L(u)$, and $G$ is \emph{degree-choosable} if it is $L$-colourable for any list assignment $L$ where, for each vertex $u$ in $G$, the length of the list~$L(u)$ is equal to the degree of $u$.  The \emph{blocks} of a graph are its maximal 2-connected subgraphs.  The following well-known fact is a special case of the characterization of degree-choosable connected graphs of Borodin~\cite{borodin} and Erd\H{o}s~et~al.~\cite{erdos1979choosability}.

\begin{lem}[\cite{borodin,erdos1979choosability}\label{lem:gallai}]
Let $G$ be a connected graph.  Then $G$ is degree-choosable unless each block of $G$ is a complete graph or an odd cycle.
\end{lem}

We require two more definitions.  Given two sets $S_1$ and $S_2$ of vertices of~$G$, we say that~$S_1$ \emph{dominates} $S_2$ if every vertex in $S_2$ is adjacent to at least one vertex in~$S_1$. Additionally, $S_1$ \emph{weakly dominates} $S_2$ if every vertex in $S_2$ is adjacent to exactly one vertex in $S_1$.  

\section{Proof of Theorem~\ref{thm:main}}\label{section:final}

We must show that, for $k \geq 4$, if $G$ is a connected non-complete $k$-regular graph, then the set of $k$-colourings of $G$ is a Kempe class.  In Propositions~\ref{prop:3connected},~\ref{cor:diam3} and~\ref{lem:complementbipartite}, we show that this claim holds, respectively, whenever~$G$ is not 3-connected, 3-connected with diameter at least 3 and with diameter exactly 2.  It is clear that taken together the propositions imply Theorem~\ref{thm:main}.

\subsection{Graphs that are not 3-connected}

We first prove that Theorem~\ref{thm:main} holds when $G$ is not 3-connected.

\begin{prop}\label{prop:3connected}
Let $k \geq 4$ be a positive integer.  Let $G$ be a connected $k$-regular graph that is not 3-connected. Then $C^k_G$ is a Kempe class.
\end{prop}

\begin{proof}
Let $S$ be a minimal vertex cut of $G$ that separates a connected component $C_1$ of $G-S$ from the rest of the graph $C_2$.  Let $G_1=G[C_1 \cup S]$ and $G_2=G[C_2 \cup S]$. Note that both $G_1$ and $G_2$ are $(k-1)$-degenerate.  Thus~$C^k_{G_1}$ and $C^k_{G_2}$ are Kempe classes by Lemma~\ref{lem:degeneratekempe}. 

If $G[S]$ is a clique, then, by Lemma~\ref{lem:cliqueintersection}, $C^k_G$ is a Kempe class.

As $G$ is not 3-connected, $|S| \leq 2$.  So if $G[S]$ is not a clique, then $S=\{x,y\}$ where $x$ and $y$ are a pair of non-adjacent vertices.  We can assume that one vertex of $S$ has more than one neighbour in $G_1$, and the other has more than one neighbour in $G_2$. Suppose instead that, for example,  $x$ and $y$ both have only one neighbour in $G_1$ (and so three neighbours in $G_2$). Note that we can, in this case, let $S$ be the cut of size~2 containing $y$ and the unique neighbour $u$ of $x$ in $G_1$. Now $S$ does have the desired property if we consider how it separates $C_1-\{u\}$ from $C_2 + \{x\}$:  $u$ has at least two neighbours distinct from $x$ and $y$ which are in $C_1-\{u\}$,  and $y$ has three neighbours in~$C_2+\{x\}$. So we can assume, without loss of generality, that~$x$ has at least two neighbours in $G_1$, and $y$ has at least two neighbours in $G_2$.

Let $G'_1$, $G'_2$ and $G'$ be the graphs obtained from, respectively, $G_1$, $G_2$ and~$G$ by adding the edge $xy$. As $x$ has degree at least 2 in $G_1$, it has degree at most $k-2$ in $G_2$, and thus degree at most $k-1$ in $G'_2$. Similarly $y$ has degree at most $k-1$ in $G'_1$. Hence $G'_1$ and $G'_2$ are $(k-1)$-degenerate and, by Lemma~\ref{lem:degeneratekempe}, $C^k_{G'_1}$ and $C^k_{G'_2}$ are Kempe classes.  By Lemma~\ref{lem:cliqueintersection}, $C^k_{G'}$ is a Kempe class.  

So the set of $k$-colourings of $G$ in which $x$ and $y$ have distinct colours are all Kempe equivalent (since this is the set of $k$-colourings of $G'$). To prove that $C^k_G$ is a Kempe class, it is enough to show that every $k$-colouring~$\alpha$ of~$G$ such that $\alpha(x)=\alpha(y)$ is Kempe equivalent to a $k$-colouring where $x$ and~$y$ are coloured differently.  We will describe how to find a series of Kempe changes that, starting from $\alpha$, give us a colouring in which $x$ and $y$ are not coloured alike.

We can assume that $\alpha(x)=\alpha(y)=1$.  If, for either $x$ or $y$, there is a colour that does not appear on any vertex in its neighbourhood then we can apply a trivial Kempe change to obtain the required colouring.  So we assume that, under $\alpha$, for each of $x$ and $y$, there is a neighbour of each colour and so exactly one colour appears on two neighbours.  We consider two cases.
 
\smallskip
\noindent {\bf Case 1:} \emph{Either $x$ or $y$ has at least two neighbours in each of $G_1$ and~$G_2$.}

\smallskip
\noindent 
Let us assume that it is $x$ that has two neighbours in both $G_1$ and $G_2$. There exist two colours --- let us say $2$ and $3$ --- such that no neighbour of $x$ in~$G_1$ is coloured $3$ and no neighbour of $x$ in $G_2$ is coloured $2$.  Consider the $(2,3)$-components of $G$ that include the neighbours of $x$ coloured 2. Since they are included in $G_1$, they do not contain any neighbour of $x$ coloured 3. So in the colouring obtained by a Kempe change of these components, the vertex~$x$ has no neighbour coloured 2. Thus by one further trivial Kempe change of~$x$, the required colouring is obtained.
 
 \smallskip
\noindent {\bf Case 2:} \emph{Neither $x$ nor $y$ has at least two neighbours in each of $G_1$ and~$G_2$.}

\smallskip
\noindent Since $S$ is minimal, $x$ has exactly one neighbour $w$ in $G_2$ and $y$ has exactly one neighbour $z$ in $G_1$. Let $\alpha(w)=2$. If $\alpha(z) \neq 2$, then consider the $(2, \alpha(z))$-component that contains $z$.    From the Kempe change of this component (which does not contain $x$, $y$ or $w$), we obtain a colouring where~$z$ is coloured~2.  Thus we can as well assume that $\alpha(z)=2$.  Consider the $(1,3)$-component that contains $x$.  As $x$ has no neighbour coloured 3 in $G_2$ and $y$ has no neighbour coloured 3 in $G_1$, this component does not contain~$y$.  Thus from the Kempe change of this component we obtain the required colouring. 
\end{proof}

\subsection{3-connected graphs with diameter at least 3}

We present a number of lemmas that will allow us to show that Theorem~\ref{thm:main} is true for 3-connected graphs with diameter at least 3.
 
If two neighbours $t_1$ and $t_2$ of a vertex $u$ are not adjacent, we say that $(t_1,t_2)$ is an \emph{eligible} pair of neighbours of $u$.  Let $P(u)$ denote the set of eligible pairs of neighbours of $u$.   We observe that in a regular connected non-complete graph, every vertex has an eligible pair of neighbours.

The next lemma follows from Lemma~\ref{lem:matching}.  (In fact, it is just a special case of Lemma~\ref{lem:matching}, but it is helpful to have it as a separate statement.)

\begin{lem}\label{lem:twomatchingneighbors}
Let $k \geq 3$ be a positive integer.  Let $G$ be a 3-connected $k$-regular graph $G$.  Let $u$ be a vertex in $G$ and let $(t_1,t_2)$ be an eligible pair in~$P(u)$. Then $C^k_G(t_1,t_2)$ is a Kempe class.  
\end{lem}

It is worth noting that as, by Lemma~\ref{lem:degenerate}, $G_{t_1+t_2}$ is $(k-1)$-degenerate, it has a $k$-colouring so $C^k_G(t_1,t_2)$ is non-empty.

\begin{lem}\label{lem:twopairsofmatchingneighbors}
Let $k \geq 4$ be a positive integer.  Let $G$ be a 3-connected \mbox{$k$-regular} graph.  Let $u$ and $v$ be two vertices of $G$ and let $(w_1,w_2)$ be an eligible pair in~$P(v)$. If, for every eligible pair~$(t_1,t_2)$ in $P(u)$, there is a $k$-colouring of~$G$ such that $t_1$ and $t_2$ are coloured alike and $w_1$ and~$w_2$ are coloured alike, then~$C^k_G$ is a Kempe class.
\end{lem} 

\begin{proof}
In a $k$-colouring of $G$ at most $k-1$ colours appear on the neighbours of $u$.  Thus at least two of its neighbours, which must be an eligible pair, are coloured alike.  That is, for every colouring $\alpha$ of $G$, there is an eligible pair~$(t_1,t_2)$ in $P(u)$ such that $\alpha$ belongs to $C^k_G(t_1,t_2)$.  So
\[
C^k_G = \bigcup_{(t_1, t_2) \in P(u)} C^k_G(t_1,t_2),
\]
and, as each $C^k_G(t_1,t_2)$ is, by Lemma~\ref{lem:twomatchingneighbors}, a Kempe class, we have that $C^k_G$ is a Kempe class if it contains a subset that is a Kempe class and intersects $C^k_G(t_1,t_2)$ for each $(t_1,t_2) \in P(u)$.  The premise of the lemma is that $C^k_G(w_1,w_2)$ intersects each $C^k_G(t_1,t_2)$, and it is, by Lemma~\ref{lem:twomatchingneighbors}, a Kempe class.
\end{proof}

So Lemma~\ref{lem:twopairsofmatchingneighbors} suggests an approach to proving that Theorem~\ref{thm:main} holds for 3-connected graphs.  We note first that it will be easier to apply if we know that $G$ has diameter at least 3, since then we can choose $u$ and $v$ such that their eligible pairs of neighbours are distinct.  We just need to prove that we can find the types of $k$-colourings that the premise of the lemma requires.  To do this we need a number of rather technical lemmas.

\begin{lem}\label{lem:no3cutcuttingaclique}
Let $k\geq 4$ be a positive integer.  Let $G$ be a $k$-regular 3-connected graph with a vertex cut $S$  of size $3$ such that one connected component $C$ of $G-S$ is a clique on $k$ vertices. Then $C^k_G$ is a Kempe class.
\end{lem}

\begin{proof}
Every vertex of $C$ has exactly one neighbour that is not in $C$ and so must therefore be in~$S$. Thus $S$ weakly dominates $V(C)$. As~$C$ has at least four vertices each adjacent to exactly one of the three vertices of $S$, we can assume that there is a vertex in $S$ with at least two neighbours in $C$.  Let this vertex be $u$. Let $w_1$ be a neighbour of $u$ in $C$.  Let $w_2$ be a neighbour of~$u$ not in $S \cup V(C)$.  (If $u$ does not have such neighbours, then $S \setminus \{u\}$ is a vertex cut and $G$ is not 3-connected.)

By Lemma~\ref{lem:twomatchingneighbors}, $C^k_G (w_1,w_2)$ is a Kempe class.  Let $\alpha$ be a $k$-colouring of $G$.  The lemma follows if we can show that $\alpha$ is Kempe equivalent to a colouring in $C^k_G (w_1,w_2)$; that is, if by performing a number of Kempe changes we can reach a colouring where $w_1$ and $w_2$ are coloured alike.

Let us assume that $\alpha(w_1)=1$.  If $\alpha(w_2)=1$, we are done so assume that $\alpha(w_2)=2$.  Let $w_3$ be the vertex in $C$ for which $\alpha(w_3)=2$ (as $C$ is a clique on $k$ vertices every colour appears on exactly one vertex).

If $w_3$ is a neighbour of $u$, then $\{w_1, w_3\}$ is a Kempe chain, and, by a single Kempe change, we obtain the required colouring.  Otherwise suppose that the neighbour of $w_3$ in $S$ is $v \neq u$.  As $u$ has at least two (distinctly coloured) neighbours in $C$, we can assume that there is a neighbour $w_4$ of $u$ in~$C$ such that $\alpha(w_4) \neq \alpha(v)$ (possibly $w_4=w_1$).  Then $\{w_3, w_4\}$ is a Kempe chain.  If we exchange the colours of this chain, then either $w_4=w_1$ and we are done or, as before, we have two neighbours of~$u$ coloured~$1$ and $2$ which form a Kempe chain, and one more Kempe change is needed to obtain the required colouring.
\end{proof}

At various points in the proofs of the following lemmas we will have defined a graph $G$ with vertices $u$ and $v$ and eligible pairs $(t_1,t_2) \in P(u)$ and $(w_1,w_2) \in P(v)$.  Whenever this is the case, we will use the following definitions.  Let $G^+$ be the graph obtained from $G$ by identifying $t_1$ and $t_2$ and then identifying $w_1$ and $w_2$, and label the two vertices created $t$ and $w$ respectively.  Let $G^-$ be the graph obtained from $G^+$ by deleting $t$ and $w$ (so~$G^-$ is the graph obtained from $G$ by deleting $t_1$, $t_2$, $w_1$ and $w_2$).  The blocks of a graph $G$ that  contain at most one cut vertex of $G$ are called end blocks, and all other blocks are called intermediate blocks.

\begin{lem}\label{lem:diam2} 
Let $k \geq 4$ be a positive integer.  Let $G$ be a 3-connected $k$-regular graph.  Let $u$ and $v$ be two vertices of $G$ and let $(w_1,w_2)$ be an eligible pair in $P(v)$ such that neither of $w_1$, $w_2$ is adjacent to $u$.   Suppose that~$C^k_G$ is not a Kempe class. Then there is an eligible pair $(t_1,t_2)$ in $P(u)$ such that~$G$ contains a subgraph weakly dominated by both $\{t_1,t_2\}$ and $\{w_1,w_2\}$ that is isomorphic to $K_{k-1}$.
 \end{lem}

\begin{proof}
As $C^k_G$ is not a Kempe class, we know, by Lemma~\ref{lem:twopairsofmatchingneighbors}, we can choose as $(t_1, t_2)$ an eligible pair in $P(u)$ such that there is no $k$-colouring of $G$ such that $t_1$ and $t_2$ are coloured alike and $w_1$ and $w_2$ are coloured alike.  We note that $t_1$, $t_2$, $w_1$ and $w_2$ are distinct as the latter two are not adjacent to $u$.  So here $G^+$ is well-defined and, by our choice of $t_1$ and $t_2$, does not have a $k$-colouring.  To prove the lemma, we attempt to construct a $k$-colouring of~$G^+$ and use the fact that we know that we cannot succeed to lead us to the conclusion.

For a component $C$ of $G^-$, let $G^\ast_C$ be $G^+[C \cup \{t,w\}]$.   For each $C$, we shall show that one of the following holds:
\begin{enumerate}[(1)]
\setlength\itemsep{0pt}
\setlength{\parskip}{0pt}
    \setlength{\parsep}{0pt} 
\item the structure of $G^\ast_C$ implies that $G^+$ has a $k$-colouring, or
\item there is a $k$-colouring of $G^\ast_C$ where $t$ and $w$ are coloured 1 and 2 respectively, or
\item $G[C]$ contains a subgraph weakly dominated, in $G$, by both $\{t_1,t_2\}$ and $\{w_1,w_2\}$ that is isomorphic to $K_{k-1}$.
\end{enumerate}
By the assumption that $G^+$ has no $k$-colouring, there cannot be any component that satisfies (1), and it cannot be the case that every component satisfies (2).  Thus there must be at least one component that satisfies (3), and the lemma follows.

\smallskip
\noindent {\bf Case 1:} There is a vertex $x$ in $C$ that has degree less than $k$ in $G^\ast_C$.

\smallskip
\noindent  We can find a $k$-colouring of $G^\ast_C$ with $t$ and $w$ coloured with 1 and 2 by applying Lemma~\ref{lem:elimination} to $G^\ast_C$ and $x$ with $S = \{t,w\}$.  So $C$ satisfies (2).

\smallskip
\noindent {\bf Case 2:} Every vertex in $C$ has degree $k$ in $G^\ast_C$ and $G[C]$ is degree-chooseable.

\smallskip
\noindent  We create a list assignment $L$ for $G[C]$.  For each vertex $x$ in $C$, let

\[  
L(x) = \left\lbrace    \begin{array}{ll}
\{1, \ldots, k\} & \mbox{if $x$ is not adjacent to $t$ or $w$,} \\
\{2, \ldots, k\} & \mbox{if $x$ is adjacent to $t$ but not $w$,} \\
\{1,3 \ldots, k\} & \mbox{if $x$ is adjacent to $w$ but not $t$,} \\
\{3 \ldots, k\} & \mbox{if $x$ is adjacent to both $t$ and $w$.} 
\end{array}
\right.
\]
Note that $|L(x)|$ is equal to the degree of $x$ in $G[C]$ since $|L(x)| = k-|N_{G^+}(x)\cap \{t,w\}|$.  As $G[C]$ is degree-chooseable, there is a colouring of $G[C]$ that respects $L$, and, as $1 \notin L(x)$ if $x$ is adjacent to $t$ and $2 \notin L(x)$ if $x$ is adjacent to $w$, this provides a $k$-colouring of $G^\ast_C$ when $t$ and $w$ are coloured~$1$ and $2$.  Thus $C$ satisfies (2).

\smallskip
\noindent {\bf Case 3:} Every vertex in $C$ has degree $k$ in $G^\ast_C$ and $G[C]$ is not degree-chooseable.

\smallskip
\noindent By Lemma~\ref{lem:gallai}, each block of $G[C]$ is either a clique or an odd cycle.  For an end block $B$ of $G[C]$, let $B^-$ be the vertices of $B$ that are not a cutvertex in~$G[C]$ (so $B^-$ contains one fewer vertex than $B$ unless $G[C]$ contains only one block and then $B^-=B$).  The degree of each vertex of $B^-$ in $G^\ast_C$ is $k$, and this is the sum of the number of neighbours it has in $C$ and the number of neighbours it has in $\{t,w\}$.  As the former is the same for each vertex (as they belong to just one block that is a cycle or a clique), the latter must also be the same for each vertex.  So let $d_B \in \{0,1,2\}$ be the number of neighbours in $\{t,w\}$ of each vertex of $B^-$.

\smallskip
\noindent {\bf Case 3.1:} There is an end block $B$ of $C$ with $d_B=0$.

\smallskip
\noindent This implies that each vertex of $B^-$ is joined to $k$ vertices in $C$ which, as $k \geq 4$, implies that $B$ is a clique rather than a cycle, and so $B$ is isomorphic to $K_{k+1}$, contradicting that $G$ is connected and non-complete.

\smallskip
\noindent {\bf Case 3.2:} There is an end block $B$ of $C$ with $d_B=1$.

\smallskip
\noindent  Note that $B$ must be a clique as if it were an odd cycle the degree of each vertex of $B^-$ in $G^\ast_C$ would be $3 \neq k$.  

Suppose every vertex in $B^-$ is adjacent to $t$ (the case where they are all adjacent to $w$ is equivalent).  We cannot have $B=B^-$ since then $t$ is a cutvertex, and so $\{t_1,t_2\}$ is a cutset in~$G$, contradicting that it is 3-connected.  So let $x$ be the cutvertex of $G[C]$ in $B$. Then $x$ has exactly one neighbour~$s$ in~$C \setminus B^-$. Thus $\{s,t_1,t_2\}$ is a vertex cut of $G$ that weakly dominates $B$ which is a clique on $k$ vertices.  Therefore $C^k_G$ is a Kempe class by Lemma~\ref{lem:no3cutcuttingaclique}; a contradiction.

So there must be vertices $y$ and $z$ in $B^-$ such that $y$ is adjacent to $t$ (but not $w$) and $z$ is adjacent to $w$ (but not $t$).  We show that we can colour $t$ and $w$ with 1 and 2, and extend this to a $k$-colouring of $G^\ast_C$.  First colour $z$ with~$1$.  Note that a vertex $h$ other than $y$ and $z$ in $B^-$ has degree less than~$k$ in $G^\ast_C \setminus \{y\}$. So we can apply Lemma~\ref{lem:elimination} to $G^\ast_C \setminus \{y\}$ with $S = \{t,w,z\}$ and $x =h$ (if $B^-$ does not contain three vertices then the degree of $y$ and~$z$ in $G^\ast_C$ is at most $3<k$).  Finally colour $y$, which is possible as two of its neighbours are coloured alike.  Thus $C$ satisfies (2).

\smallskip
\noindent For the remaining cases, we will need the following claim.

\begin{claim} \label{claim1}
If $u$ and $v$ are not in $C$, then
\begin{enumerate}[\normalfont A.]
\setlength\itemsep{0pt}
\setlength{\parskip}{0pt}
    \setlength{\parsep}{0pt} 
\item each of $t$ and $w$ is adjacent to at most $2k-2$ vertices in $C$,
\item one of $t$ and $w$ is adjacent to at most $2k-3$ vertices in $C$, 
\item if each of $t$ and $w$ has at least $2k-3$ neighbours in $C$, then $t$ is not adjacent to $w$, and
\item if the sum of the number of neighbours of $t$ and $w$ in $C$ is at least $4k-6$, then $G^+[V \setminus C]$ has a $k$-colouring in which $t$ and $w$ are coloured alike.
\end{enumerate}
\end{claim}

\noindent We note that this claim can be applied within Case 3; we know that every vertex in $C$ has degree $k$ in $G^\ast_C$, and $u$ and $v$ have degree less than $k$ since a pair of neighbours --- $t_1$ and $t_2$ or $w_1$ or $w_2$ --- were identified when $G^+$ was formed from $G$.  We prove each part of the claim (we give a proof only for the statement about $t$ when the argument for $w$ is equivalent).    We keep in mind that, for each edge incident with $t$ in $G^+$, there is a corresponding edge or edges incident with $t_1$ or $t_2$ in $G$.
\begin{enumerate}[A.]
\setlength\itemsep{0pt}
\setlength{\parskip}{0pt}
    \setlength{\parsep}{0pt} 
\item The total number of edges incident with $t_1$ and $t_2$ in $G$ is $2k$, but 2 of these are incident with $u$ which is not in $C$.
\item If $t$ and $w$ both have $2k-2$ neighbours in $C$, then in $G$, each of $t_1$, $t_2$, $w_1$ and $w_2$ only has neighbours in $C \cup \{u,v\}$.  Then $\{u,v\}$ is a cutset as it separates $C \cup \{t_1,t_2, w_1, w_2\}$ from the rest of $G$ which is not empty as $u$ has at least~4 neighbours and is not adjacent to any vertex in $C \cup \{w_1, w_2\}$.  This contradicts that $G$ is 3-connected.
\item If $t$ and $w$ both have $2k-3$ neighbours in $C$ and are adjacent, then, in $G$, each of $t_1$, $t_2$, $w_1$ and $w_2$ only has neighbours in $C \cup \{u,v, t_1, t_2, w_1, w_2\}$, and, as in the previous part, this implies that $\{u,v\}$ is a cutset.
\item We can say that $t$ and $w$ are not adjacent: either one of $t$ and $w$ has $2k-2$ neighbours in $C$ so its only other neighbour is either $u$ or $v$, or they both have $2k-3$ neighbours in $C$ and so we can apply the previous part of the claim.  In $G$, there are at least $4k-6$ edges from $\{t_1,t_2, w_1, w_2\}$ to the vertices of $C$ so at most 6 other incident edges.  And, as $t_1$ and $t_2$ are both adjacent to $u$, and $w_1$ and $w_2$ are both adjacent to $v$, in $G^+[V \setminus C]$ the sum of the degrees of $t$ and $w$ is at most 4.  Let $G^\dagger$ be the graph formed from $G^+[V \setminus C]$ by identifying $t$ and $w$ to form a new vertex with degree at most 4.  Thus every vertex in $G^\dagger$ has degree at most $k$, and the graph is not isomorphic to $K_{k+1}$ (since $u$, for example, has degree less than $k$) so, by Brooks' Theorem, $G^\dagger$ has a $k$-colouring.  From this colouring, we can obtain a colouring of $G^+[V \setminus C]$ in which $t$ and $w$ are coloured alike.  This completes the proof of the claim.
\end{enumerate}

\smallskip
\noindent {\bf Case 3.3:} For every end block $B$ of $C$, $d_B=2$, and there is one end block~$B_1$ that is not a clique.

\smallskip
\noindent So $B_1$ is an odd cycle on at least five vertices.  In $G^\ast_C$, each vertex of $B_1^-$ has degree $k$ and is adjacent to two vertices in $B_1$ and $t$ and $w$. Hence we have $k=4$.  If either $B_1$ has more than five vertices or $C$ has more than one end block, then there are at least six vertices in end blocks that are not cutvertices of~$C$. So these vertices are adjacent to both $t$ and $w$, which therefore each have at least $6=2k-2$ neighbours in $C$, contradicting Claim~\ref{claim1}.B.  So $C=B_1$ is a 5-cycle, and the sum of the number of neighbours of $t$ and $w$ in $C$ is $10=4k-6$.   Thus, by Claim~\ref{claim1}.D, $G^+[V \setminus C]$ has a 4-colouring in which $t$ and $w$ are coloured alike.  We  can extend this colouring to the whole of $G^+$ by using the other 3 colours on $B_1$.  So $C$ satisfies (1).

\smallskip
\noindent {\bf Case 3.4:} For every end block $B$ of $C$, $d_B=2$ and $B$ is a clique.

\smallskip
\noindent Notice that each end block is isomorphic to $K_{k-1}$.  If there is only one end block, then, as it is weakly dominated by $\{t_1,t_2\}$ and $\{w_1, w_2\}$ in $G$, $C$ satisfies (3).  If there are at least three end blocks, then there are at least $3(k-2)$ vertices in $C$ adjacent to both $t$ and $w$.  As, for $k \geq 4$, we have $3k-6 \geq 2k-2$, this contradicts Claim~\ref{claim1}.B.

So we can assume that $C$ has exactly two end blocks each isomorphic to~$K_{k-1}$.  Note that an  intermediate block $B$ of $C$ that is a clique on more than two vertices has vertices (the ones that are not cutvertices in $G[C]$) whose $k$ neighbours are each either in $B$ or in $\{t, w\}$.  In fact, at least one neighbour must be in $\{t,w\}$, else $B$ is isomorphic to $K_{k+1}$ and not connected to the rest of $G$. Therefore $B$ is isomorphic to either $K_{k-1}$ or $K_k$.  

\smallskip
\noindent {\bf Case 3.4.1:} $k \geq 5$.

\smallskip
\noindent
No block is an odd cycle (since the vertices that are not cutvertices in the cycle would have degree at most 4 in $G^\ast_C$).  So the blocks of $C$ are each isomorphic to $K_2$, $K_{k-1}$ or $K_k$ and, for each cutvertex, one of the two blocks it belongs to must be $K_2$ else it would have degree at least $2(k-2)>k$.  Thus the cutvertex of each end block is also adjacent to one of $t$ and $w$, so there are $4k-6$ edges from $t$ and $w$ to vertices of the two end blocks.  If there is an intermediate block that is isomorphic to $K_{k-1}$ or $K_k$, then it contains at least two vertices that are not cutvertices, and these are also joined to at least one of $t$ and $w$.  So the sum of the number of neighbours of $t$ and $w$ in $C$ is at least $4k-4$; a contradiction to the first two parts of Claim~\ref{claim1}.  Therefore the only intermediate block is $K_2$, and there is exactly one of these (if there are none, the two end blocks intersect and the cutvertex has degree more than $k$; if there is more than 1, there are vertices that in $G[C]$ have degree 2 so have degree at most 4 in $G^\ast_C$).  So $G[C]$ contains two disjoint cliques each isomorphic to $K_{k-1}$ joined by a single edge.  Thus the sum of the number of neighbours of $t$ and $w$ in $C$ is exactly $4k-6$, and we can assume, by Claim~\ref{claim1}.D, that $G^+[V \setminus C]$ has a $k$-colouring in which $t$ and $w$ are coloured alike.  This can be extended to a colouring of $G^+$ as $G[C]$ is easily seen to be $(k-1)$-colourable.  So $C$ satisfies (1).

\smallskip
\noindent {\bf Case 3.4.2:} $k=4$.

\smallskip
\noindent  Let the two end blocks be $B_1$ and $B_2$ (both are isomorphic to $K_3$).  If they intersect in a vertex, then we can colour $t$ and $w$ with 1 and 2, colour the vertex in both $B_1$ and $B_2$ with 1 and the other vertices with 3 and 4.  So $C$ satisfies (2).

For the remaining cases, we  note that Claim~\ref{claim1}.D says that if there are at least 10 edges joining $t$ and $w$ to $C$ we can assume they are coloured alike in a $4$-colouring of $G^+[V \setminus C]$. And Claim~\ref{claim1}.A and ~\ref{claim1}.B say that there cannot be more than 11 edges from $t$ and $w$ to $C$.

If $G[C]$ is $B_1$ and $B_2$ plus an edge between them, then there are 10 edges from $t$ and $w$ to $C$, and clearly $G[C]$ is 3-colourable so $C$ satisfies (1) by Claim~\ref{claim1}.D. 

Suppose that $G[C]$ contains more blocks than $B_1$, $B_2$ and an additional~$K_2$.  If $C$ does not contain a $K_4$, then either there is a block isomorphic to $K_3$ or a longer odd cycle that contains a vertex $x$ that is not a cutvertex, or there is a cutvertex $x$ that belongs to two blocks both isomorphic to $K_2$.  In both cases, $x$ must be joined to both $t$ and $w$ which are therefore again joined by at least 10 edges to $C$, and, as there is no $K_4$, $G[C]$ is 3-colourable and, by Claim~\ref{claim1}.D, $C$ again satisfies~(1). 

If $C$ does contain a $K_4$, then the two vertices that are not cutvertices are both incident to one of $t$ and $w$.   And the cutvertices in $B_1$ and $B_2$ are each either adjacent to one of $t$ or $w$, or belong to a $K_3$ or a longer odd cycle that contains a vertex adjacent to both $t$ and $w$.  In any case, $t$ and $w$ are incident to at least 12 edges joining them to $C$, and we have a contradiction.
\end{proof}

\begin{lem}\label{lem:diam3}
Let $k \geq 4$ be a positive integer.  Let $G$ be a 3-connected $k$-regular graph. Let $u$ and $v$ be two vertices of $G$ that are not adjacent.   Let $(w_1,w_2)$ be an eligible pair in $P(v)$ neither of which is adjacent to $u$.  Then~$C^k_G$ is a Kempe class.
\end{lem}

\begin{proof}
If $C^k_G$ is not a Kempe class, then, by Lemma~\ref{lem:diam2}, there is an eligible pair $(t_1,t_2)$ in $P(u)$ such that $G$ contains an induced subgraph isomorphic to~$K_{k-1}$ that is weakly dominated by both $\{t_1, t_2\}$ and $\{w_1, w_2\}$.  Let $C$ be the vertex set of this induced subgraph. Note that each vertex in $C$ is adjacent to the other $k-2$ vertices of $C$, to one of $\{t_1, t_2\}$ and to one of $\{w_1, w_2\}$, and so is not adjacent to $u$ or $v$ (neither of which can be in $C$ as they are each adjacent to both of the vertices in either $\{t_1, t_2\}$ or $\{w_1, w_2\}$).  We can assume that each of $\{t_1, t_2, w_1, w_2\}$ is adjacent to at least one vertex in $C$: if fewer than three of them have a neighbour in $C$, then $G$ is not 3-connected, and if exactly one of them, say $t_1$, has no neighbour in $C$, then, since $C \cup t_2$ would induce a clique on $k$ vertices that is weakly dominated by $\{u, w_1, w_2\}$ (every vertex in $C$ is adjacent to one of $\{w_1, w_2\}$ but not to $u$ and $t_2$ is adjacent to~$u$ but, considering its degree, not to either of $\{w_1, w_2\}$), and Lemma~\ref{lem:no3cutcuttingaclique} with $S = \{u, w_1, w_2\}$ is contradicted.

Assume, without loss of generality, that $w_1$ has at least as many neighbours in $C$ as $w_2$.  Let $x$ be a neighbour of $w_1$ in $C$, and assume, without loss of generality, that $x$ is also a neighbour of $t_1$.  Then $(x,v)$ is an eligible pair in~$P(w_1)$.  We apply Lemma~\ref{lem:diam2} to $u$, $w_1$ and $(x, v)$.  So, under the assumption that $C^k_G$ is not a Kempe class, there is a pair $(t_3,t_4)$ (not necessarily distinct from $(t_1,t_2)$) of eligible neighbours in $P(u)$ such that $G$ contains an induced subgraph isomorphic to $K_{k-1}$ that is weakly dominated by both~$\{t_3, t_4\}$ and~$\{x, v\}$.  Let $C'$ be the vertex set of this induced subgraph and, arguing as we did for $C$, we can assume that each of $\{t_3, t_4, x, v\}$ is adjacent to $C'$.

Suppose that neither $t_1$ nor $t_2$ belongs to $C'$.   The $k$ neighbours of $x$ are $C \setminus \{x\} \cup \{t_1, w_1\}$, and we know at least one of these vertices is in $C'$.  By definition, it is not $w_1$, and, by assumption, it is not $t_1$.  So there is a vertex $y \neq x$ that belongs to both $C$ and $C'$.  As $C'$ induces a clique, the other $k-2$ vertices of $C'$ are neighbours of $y$.  But, as none of $\{t_1, t_2, w_1, x\}$ are in $C'$, we must have that $C'$ is $C \setminus \{x\} \cup \{w_2\}$.  So $w_2$ is adjacent to every vertex of $C$ except $x$.  By our assumption that $w_1$ has at least as many neighbours as $w_2$ in $C$, we have that $C$ has only two vertices and so $k=3$.  This contradiction tells us that, in fact, at least one of $t_1$ and $t_2$ belongs to $C'$; let us assume it is $t_1$.

So $t_1$ has $k-2$ neighbours in $C'$.  It has two more neighbours: we know it must be adjacent to one of $\{v,x\}$ by the definition of $C'$, and we know that it is also adjacent to $u$.  But neither of $t_3$ and $t_4$ belongs to $C' \cup \{u,v,x\}$ so~$t_1$ is not adjacent to either of them.  This contradicts the definition of $C'$ and completes the proof.
\end{proof}

We can now conclude this subsection on graphs of diameter at least 3.

\begin{prop}\label{cor:diam3}
Let $k \geq 4$ be a positive integer.  Let $G$ be a $3$-connected $k$-regular graph with diameter at least $3$. Then $C_G^k$ is a Kempe class. 
\end{prop}

\begin{proof}
Let $u$ and $v$ be two vertices in $G$ at distance at least $3$.   As $G$ has diameter 3, it is not a clique, and so, as it is regular, $v$ has a pair of non-adjacent neighbours, an eligible pair.  As no neighbour of $v$ is adjacent to $u$ (they must be at distance at least 2), the result follows from Lemma~\ref{lem:diam3}. 
\end{proof}

\subsection{3-connected graphs with diameter 2}

To complete the proof of Theorem~\ref{thm:main}, it only remains to consider 3-connected graphs of diameter 2.  

First we need two definitions.  For a vertex $v$ in a graph $G$, we denote by~$N(v)$ the \emph{neighbourhood} of $v$, that is, the set of vertices adjacent to $v$. The \emph{second neighbourhood} of $v$ is the set of vertices at distance 2 from $v$ in~$G$.

\begin{prop}\label{lem:complementbipartite}
Let $k \geq 4$ be a positive integer.  Let $G$ be a $3$-connected $k$-regular graph of diameter $2$. Then $C_G^k$ is a Kempe class. 
\end{prop}
\begin{proof}  
If there are three vertices in the second neighbourhood of a vertex $v$ that induce a path, then the proposition follows immediately from Lemma~\ref{lem:diam3}. Therefore we can assume that the second neighbourhood of each vertex induces a disjoint union of cliques.

Assume that there is a vertex $v$ whose second neighbourhood includes two sets of vertices $C_1$ and $C_2$ that each induce a clique, and that there are no edges between the two sets.  Let $x$ and $y$ be vertices of $C_1$ and $C_2$ respectively.  If~$x$ is adjacent to a neighbour $z$ of $v$ that is not adjacent to $y$, then the second neighbourhood of $y$ contains an induced path on $v$, $z$ and $x$ and, again, we are done by Lemma~\ref{lem:diam3}.  Thus, by symmetry, the intersections of each of the neighbourhoods of $x$ and $y$ with $N(v)$ are the same and, repeating the argument, we must have that every vertex of $C_1$ and $C_2$ has the same set of neighbours within $N(v)$.  Let $\alpha$ be a $k$-colouring of $G$.  Suppose that $\alpha(x)=1$ and $\alpha(y)=2$. Note that the $(1,2)$-component that contains $x$ contains only vertices of $C_1$. Exchange the colours on this $(1,2)$-component and let $\beta$ be the resulting colouring. So $\beta(x)=\beta(y)=2$.  Thus, from any $k$-colouring, we can obtain by a single Kempe change a colouring in $C^k_G(x,y)$.  The proposition follows from Lemma~\ref{lem:matching}.  

Therefore we can assume that the second neighbourhood of each vertex induces a single clique.  Let $\alpha$ and $\beta$ be two $k$-colourings of $G$.  Let $v$ be a vertex, and let us denote by $C$ the second neighbourhood of $v$. Up to a Kempe change, we can assume that $\alpha(v) = \beta(v) = 1$.   To complete the proof, we assume that $\alpha$ and $\beta$ are not Kempe equivalent, and show that this leads to a contradiction.

\begin{claim} \label{secondclaim}
Neither $\alpha$ nor $\beta$ is Kempe equivalent to a colouring $\gamma$ such that $\gamma(v)=1$ and the colour 1 is not used on $C$.
\end{claim}

Suppose that there is such a colouring $\gamma$ that is Kempe equivalent to, say,~$\alpha$.  Let $x$ be the vertex in $C$ with $\beta(x)=1$ if such a vertex exists; otherwise let $x$ be any vertex in $C$.  In $\gamma$, $v$ is the only vertex in $G$ coloured~1 (since certainly there is no vertex in $N(v)$ coloured 1), so we can apply a trivial Kempe change to $x$ from $\gamma$ to obtain a colouring $\gamma'$ where $\gamma'(x)=1$.  If no vertex in $\beta$ is coloured 1, then we can use the same argument; that is, apply a trivial Kempe change to $x$ to obtain a colouring where $x$ is coloured~1.  So we may as well assume that $\beta(x)=1$, and thus, as $v$ and $x$ are coloured~1 in both $\gamma'$ and $\beta$, we have, by   Lemma~\ref{lem:matching} that $\gamma'$ and $\beta$, and so also $\alpha$ and~$\beta$, are Kempe equivalent; a contradiction that proves the claim.
  
  \smallskip

One thing that Claim~\ref{secondclaim} tells us is that $\alpha$ and $\beta$ are colourings where the colour 1 is used on $C$.  So let $u$ and $w$ be vertices in $C$ such that $\alpha(u)=1$ and $\beta(w)=1$.  If $u = w$, then Lemma~\ref{lem:matching} implies that $\alpha$ and $\beta$ are Kempe equivalent.  So, by assumption, we have $u \neq w$.

One more definition is required: given a colouring $\gamma$, a vertex $x$ is \emph{locked} if all the colours distinct from $\gamma(x)$ appear in its neighbourhood.  Notice that if $x$ is not locked, then we can apply a trivial Kempe change to $x$ from $\gamma$.

\begin{claim}\label{thirdclaim} 
Each vertex in $u \cup N(u) \setminus w$ is locked in $\alpha$. Moreover, only colour~$\alpha(w)$ appears twice in the neighbourhood of $u$.
\end{claim}
   First consider the $(1,\alpha(w))$-component of $\alpha$ containing $u$ and $w$. If this component does not contain $v$, then the Kempe change of this component from $\alpha$ gives us a colouring in which $w$ and $v$ are both coloured 1. By Lemma~\ref{lem:matching}, this colouring is Kempe equivalent to $\beta$, a contradiction. Thus $v$ must be in the $(1,\alpha(w))$-component.   Since no other neighbour of $w$ distinct from $u$ is coloured $1$ (every vertex in $G$ is a neighbour of $v$ or $u$), another neighbour $y$ of $u$ must be coloured with $\alpha(w)$.   If $u$ is not locked, then a trivial Kempe change of $u$ gives us a colouring in which 1 is not used on~$C$, contradicting Claim~\ref{secondclaim}.   Thus all the colours appear exactly once on the neighbourhood of $u$ except colour $\alpha(w)$ which appears twice.  If~$y$ is not locked, then a trivial Kempe change of $y$ returns us to the case where the $(1,\alpha(w))$-component of $\alpha$ containing $u$ and $w$ does not contain $v$.  And if a neighbour $z$ of $u$ not in $\{w ,y \}$ is not locked, then a trivial Kempe change of~$z$ returns us to the case where $u$ is not locked. The claim is proved.

\smallskip
\noindent {\bf Case 1:} \emph{$|C|\geq 3$.}

\smallskip
\noindent  Let $z \in C \setminus \{u, w\}$.  Clearly $u$ is the unique neighbour of $z$ coloured with $1$ in~$\alpha$ (since, again, every in $G$ is a neighbour of $v$ or $u$).  Similarly $w$ is the unique neighbour of $z$ coloured with $1$ in $\beta$. By Claim~\ref{thirdclaim}, $z$ is the unique neighbour of $u$ coloured $\alpha(z)$, and so $\{u,z\}$ is a Kempe chain in $\alpha$. Similarly,  noting that Claim~\ref{thirdclaim} also holds for $\beta$ with the roles of $u$ and $w$ interchanged, $\{w,z\}$ is a Kempe chain in $\beta$. By exchanging the colours on these Kempe chains, we obtain two colourings where $v$ and $z$ are each coloured 1.  Lemma~\ref{lem:matching} then implies that $\alpha$ and $\beta$ are Kempe equivalent, a contradiction.
 
\smallskip
\noindent {\bf Case 2:} \emph{$|C|=2$.}

\smallskip
\noindent  
So $G$ contains $v$, its $k$ neighbours, and $u$ and $w$.  Each of $u$ and $w$ are adjacent to all but one of the neighbours of $v$, so at least $k-2$ of the neighbours of~$v$ are adjacent to both $u$ and $w$; let this set of neighbours be denoted $S$. By Claim~\ref{thirdclaim}, in $\alpha$ a common neighbour $z$ of $u$ and $v$ is coloured $\alpha(w)$, so it follows that $S$ contains exactly $k-2$ vertices as it cannot contain $z$. As each vertex of $S$ is locked in $\alpha$ by Claim~\ref{thirdclaim} and has two neighbours, $u$ and $v$, coloured 1, they each have exactly one neighbour of each other colour.  Thus as $w$ and $z$ are coloured alike, and every vertex in $S$ is adjacent to $w$, no vertex in $S$ is adjacent to $z$.  But then the only vertices that can be adjacent to $z$ are $u$, $v$ and the other neighbour of $v$ that is not in $S$, which contradicts that $k\geq 4$. This completes Case 2, and the proof of the proposition.
\end{proof}

\section{Final Remarks}\label{sec:final}

Let us note an immediate corollary of the results on regular graphs.

\begin{corollary}\label{cor:final}
Let  $d$ and $k$ be positive integers, $d \geq k \geq 3$, and let $G$ be a connected graph with maximum degree $k$. If $G$ has a $k$-colouring, then $C_G^d$ is a Kempe class unless $d = k = 3$ and $G$ is the triangular prism.
\end{corollary}

\begin{proof}
A connected graph with maximum degree $k$ that is $k$-colourable is either $k$-regular, or $(k-1)$-degenerate but not complete.  The corollary thus follows from Theorem~\ref{thm:main} and~\cite[Theorem 1]{FJP15} and Lemma~\ref{lem:degeneratekempe} which prove the result for, respectively, regular graphs of degree greater than 3, 3-regular graphs, and $(k-1)$-degenerate graphs.
\end{proof}

It would also be interesting to find an upper bound on the diameter of~$\mathcal{K}_k(G)$ for graphs of bounded maximum degree.  In the case where the edge relation of the reconfiguration graph represents trivial Kempe changes,  this question has been well-studied; for example, see~\cite{BB13, BJLPP14, BP14, CHJ06b, FJP14}.  Notice that from the proofs of Theorem~\ref{thm:main} and~\cite[Theorem 1]{FJP15}, it suffices to determine an upper bound on the diameter of $\mathcal{K}_{d+1}(H)$, where $d\geq 2$ and $H$ is a $d$-degenerate graph, and we note that a trivial upper bound of~$O(d^{|V(H)|})$ can be easily derived from the proof of Lemma~\ref{lem:degeneratekempe}. We make the following conjecture in light of Cereceda's conjecture~\cite{luisthesis} on the diameter of reconfiguration graphs of colourings of graphs with bounded degeneracy.

\begin{conjecture}
Let $k$ and $d$ be positive integers such that $k \geq d+1$, and let $G$ be a $d$-degenerate graph on $n$ vertices. Then $\mathcal{K}_k(G)$ has diameter $O(n^2)$. 
\end{conjecture} 

\section*{Acknowledgements}

The authors are indebted to Jes\'{u}s Salas for pointing out the implications of Theorem~\ref{thm:main} on the ergodicity  of the WSK algorithm. We thank Alan Sokal and two anonymous reviewers whose comments have helped to improve the presentation of the paper.

\bibliography{bibliography}{}
\bibliographystyle{abbrv}

\end{document}